%% file: proof.tex
\documentclass{article}
\usepackage[utf8]{inputenc}
\usepackage{amsthm}
\usepackage{amsmath,amssymb}
\usepackage{mathtools}
\usepackage{enumitem}
\usepackage{tikz}

\usepackage{subfig}
\usepackage{listings}

\input{tikzstyle.tikzstyles}

\usetikzlibrary{automata}

\newtheorem{definition}{Definition}
\newtheorem{theorem}{Theorem}[section]

\newtheorem{lemma}[theorem]{Lemma}
\newtheorem{proposition}[theorem]{Proposition}

\newenvironment{roster}
 {\begin{enumerate}[font=\upshape,label=(\alph*)]}
 {\end{enumerate}}

\newcommand{\N}{\mathbb{N}}
\newcommand{\lcm}{\mathrm{lcm}}
\newcommand{\GenSubset}{\mathrm{GenSubset}}
\newcommand{\GenState}{\mathrm{GenState}}
\newcommand{\Testing}{\mathrm{Testing}}
\newcommand{\gst}[1]{\mathsf{#1}}

\title{Descriptional Complexity of Winning Sets of Regular Languages}

\author{
  Pierre Marcus \\
  ENS Lyon, Lyon, France
  \and
  Ilkka T\"orm\"a\thanks{Author supported by Academy of Finland grant 295095.} \\
  Department of Mathematics and Statistics, \\
  University of Turku, Turku, Finland \\
  \texttt{iatorm@utu.fi}
}

\begin{document}

\maketitle

\begin{abstract}
We investigate certain word-construction games with variable turn orders.
In these games, Alice and Bob take turns on choosing consecutive letters of a word of fixed length, with Alice winning if the result lies in a predetermined target language.
The turn orders that result in a win for Alice form a binary language that is regular whenever the target language is, and we prove some upper and lower bounds for its state complexity based on that of the target language.
\end{abstract}

\section{Introduction}
\label{sec:Intro}

Let us define a word-construction game of two players, Alice and Bob, as follows.
Choose a number $n \in \N$, a set of binary words $T \subseteq \{0, 1\}^n$ called the \emph{target set} and a word $w \in \{A, B\}^n$ called the \emph{turn order}, where $A$ stands for Alice and $B$ for Bob.
The players construct a word $v \in \{0, 1\}^n$ so that, for each $i = 0, 1, \ldots, n-1$ in this order, the player specified by $w_i$ chooses the symbol $v_i$.
If $v \in T$, then Alice wins the game, and otherwise Bob wins.
The existence of a winning strategy for Alice depends on both the target set and the turn order.
We fix the target set $T$ and define its \emph{winning set} $W(T)$ as the set of those words over $\{A,B\}$ that result in Alice having a winning strategy.
We extend this definition to languages $L \subseteq \{0, 1\}^*$ by considering each length separately, so that $W(L) \subseteq \{A, B\}^*$ can also contain words of variable lengths.

Winning sets were defined under this name in~\cite{salo2014playing} in the context of symbolic dynamics, but they have been studied before that under the name of \emph{order-shattering sets} in~\cite{Anstee2002,Friedl2003}.
The winning set has several interesting properties: it is downward closed in the index-wise partial order induced by $A < B$ (as changing an $A$ to a $B$ always makes the game easier for Alice) and it has the same cardinality as the target set.
This latter property was used in~\cite{peltomaki2019winning} to study the growth rates of substitutive subshifts.

If the target language $L$ is regular, then so is $W(L)$, as it can be recognized by an alternating finite automaton, which only recognizes regular languages~\cite{Chandra1981}.
Thus we can view $W$ as an operation on the class of regular languages, and in this article we study its state complexity in the general case and in several subclasses.
In our construction the alternating automaton has the same state set as the original DFA, so our setting resembles parity games, where two players construct a path in a finite automaton~\cite{Zielonka1998}.
The main difference is that in a parity game, the player who chooses the next move is the owner of the current state, whereas in our word-construction game it is determined by the turn order word.

In the general case, the size of the minimal DFA for $W(L)$ can be doubly exponential in that of $L$.
We derive a lower, but still superexponential, upper bound for bounded regular languages (languages that satisfy $L \subseteq w_1^* w_2^* \cdots w_k^*$ for some words $w_i$).
We also study certain bounded permutation invariant languages, where membership is defined only by the number of occurrences of each symbol.
In particular, we explicitly determine the winning sets of the languages $L_k = (0^* 1)^k 0^*$ of words with exactly $k$ occurrences of $1$.

In this article we only consider the binary alphabet, but we note that the definition of the winning set can be extended to languages $L \subseteq \Sigma^*$ over an arbitrary finite alphabet $\Sigma$ in a way that preserves the properties of downward closedness and $|L| = |W(L)|$.
The turn order word is replaced by a word $w \in \{1, \ldots, |\Sigma|\}^*$.
On turn $i$, Alice chooses a subset of size $w_i$ of $\Sigma$, and Bob chooses the letter $v_i$ from this set.

\section{Definitions}

We present the standard definitions and notations used in this article.
For a set $\Sigma$, we denote by $\Sigma^*$ the set of finite words over it, and the length of a word $w \in \Sigma^n$ is $|w| = n$.
The notation $|w|_a$ means the number of occurrences of symbol $a \in \Sigma$ in $w$.
The empty word is denoted by $\lambda$.
For a language $L \subseteq \Sigma^*$ and $w \in \Sigma^*$, denote $w^{-1} L = \{ v \in \Sigma^* \;|\; w v \in L \}$.

A finite state automaton is a tuple $\mathcal{A} = (Q, \Sigma, q_0, \delta, F)$ where $Q$ is a finite state set, $\Sigma$ a finite alphabet, $q_0 \in Q$ the initial state, $\delta$ is the transition function and $F \subseteq Q$ is the set of final states.
The language accepted from state $q \in Q$ is denoted $\mathcal{L}_q(\mathcal{A}) \subseteq \Sigma^*$, and the language of $\mathcal{A}$ is $\mathcal{L}(\mathcal{A}) = \mathcal{L}_{q_0}(\mathcal{A})$.
The type of $\delta$ and the definition of $\mathcal{L}(\mathcal{A})$ depend on which kind of automaton $\mathcal{A}$ is.
\begin{itemize}
\item
  If $\mathcal{A}$ is a deterministic finite automaton, or DFA, then $\delta : Q \times \Sigma \to Q$ gives the next state from the current state and an input symbol.
  We extend it to $Q \times \Sigma^*$ by $\delta(q, \lambda) = q$ and $\delta(q, s w) = \delta(\delta(q, s), w)$ for $q \in Q$, $s \in \Sigma$ and $w \in \Sigma^*$.
  The language is defined by $\mathcal{L}_q(\mathcal{A}) = \{ w \in \Sigma^* \;|\; \delta(q, w) \in F \}$.
\item
  If $\mathcal{A}$ is a nondeterministic finite automaton, or NFA, then $\delta : Q \times \Sigma \to 2^Q$ gives the set of possible next states.
  We extend it to $Q \times \Sigma^*$ by $\delta(q, \lambda) = \{q\}$ and $\delta(q, s w) = \bigcup_{p \in \delta(q, s)} \delta(p, w)$ for $q \in Q$, $s \in \Sigma$ and $w \in \Sigma^*$.
  The language is defined by $\mathcal{L}_q(\mathcal{A}) = \{ w \in \Sigma^* \;|\; \delta(q, w) \cap F \neq \emptyset \}$.
\end{itemize}
An NFA can be converted into an equivalent DFA by the standard subset construction. 

Two states $p, q \in Q$ of $\mathcal{A}$ are equivalent, denoted $p \sim q$, if $\mathcal{L}_p(\mathcal{A}) = \mathcal{L}_q(\mathcal{A})$.
Every regular language $L \subseteq \Sigma^*$ is accepted by a unique DFA with the minimal number of states, which are all nonequivalent, and every other DFA that accepts $L$ has an equivalent pair of states.
Two words $v, w \in \Sigma^*$ are congruent by $L$, denoted $v \equiv_L w$, if for all $u_1, u_2 \in \Sigma^*$ we have $u_1 v u_2 \in L$ iff $u_1 w u_2 \in L$.
They are right-equivalent, denoted $v \sim_L w$, if for all $u \in \Sigma^*$ we have $v u \in L$ iff $w u \in L$.
The set of equivalence classes $\Sigma^* / {\equiv_L}$ is the syntactic monoid of $L$, and if $L$ is regular, then it is finite.
In that case the equivalence classes of $\sim_L$ can be taken as the states of the minimal DFA of $L$.

Let $\mathcal{P} : 2^{\Sigma^*} \to 2^{\Sigma^*}$ be an operation on languages, which may not be defined everywhere.
The (regular) state complexity of $\mathcal{P}$ is the function $f : \N \to \N$, where $f(n)$ is the maximal number of states in a minimal automaton of $\mathcal{P}(\mathcal{L}(\mathcal{A}))$ for an $n$-state DFA $\mathcal{A}$.

\section{Winning Sets}

In this section we define winning sets of binary languages, present the construction of the winning set of a regular language using alternating automata, and prove some general lemmas.
We defined the winning set informally at the beginning of Section~\ref{sec:Intro}.
Now we give a more formal definition which does not explicitly mention games.

\begin{definition}[Winning Set]
  Let $n \in \N$ and $T \subseteq \{0,1\}^n$ be arbitrary.
  The \emph{winning set} of $T$, denoted $W(T) \subseteq \{A, B\}^n$, is defined inductively as follows.
  If $n = 0$, then $T$ is either the empty set or $\{\lambda\}$, and $W(T) = T$.
  If $n \geq 1$, then $W(T) = \{ A w \;|\; w \in W(0^{-1} T) \cup W(1^{-1} T) \} \cup \{ B w \;|\; w \in W(0^{-1} T) \cap W(1^{-1} T) \}$.
  
  For a language $L \subseteq \{0, 1\}^*$, we define $W(L) = \bigcup_{n \in \N} W(L \cap \{0,1\}^n)$.
\end{definition}

The idea is that for Alice to win on a turn order of the form $A w$, she has to choose either $0$ or $1$ as the first letter $v_0$ of the constructed word $v$, and then follow a winning strategy on the target set $v_0^{-1} T$ and turn order $w$.
On a word $B w$, Alice should have a winning strategy on $v_0^{-1} T$ and $w$ no matter which letter Bob chooses as $v_0$.

In the next result, a language $L$ over a linearly ordered alphabet $\Sigma$ is \emph{downward closed} if $v \in L$, $w \in \Sigma^{|v|}$ and $w_i \leq v_i$ for each $i = 0, \ldots, |v|-1$ always implies $w \in L$.

\begin{proposition}[Propositions~3.8 and~5.4 in~\cite{salo2014playing}]
    The winning set $W(L)$ is downward closed (with respect to the order $A < B$) and satisfies $|W(L)| = |L|$.
    If $L$ is a regular language, then $W(L)$ is also regular.
\end{proposition}

From an DFA $\mathcal{A}$, we can easily construct an alternating automaton for $W(\mathcal{A})$, with the same states. $B$ letters are handled with universal
transitions and $A$ with existential transitions. We don't give an explicit construction of this alternating automaton, but we work
on a corresponding NFA described in the next definition.

\begin{definition}[Winning Set Automaton]
    \label{winningSetAuto}
    Let $\mathcal{A} = (Q, \{0,1\}, q_0, \delta, F)$ be a binary DFA.
    We define a ``canonical'' NFA for $W(\mathcal{L}(\mathcal{A}))$ as follows. The states are subsets of $Q$. From a state $S \subseteq Q$, reading $B$ leads to 
    the set containing all the successors in $\mathcal{A}$ of elements of $S$. Reading $A$ leads nondeterministically to all 
    sets containing for each element of $S$, either its successor when reading $0$, or the one when reading $1$.
    The only initial state is $\{q_0\}$, and final states are all subsets of $F$.


    We usually work on the determinization of this NFA, which we denote by $W(\mathcal{A}) = (2^{2^Q}, \{A, B\}, \{\{q_0\}\}, F_W, \delta_W)$.
    Here $F_W = \{ \gst{S} \in 2^{2^Q} \;|\; \exists S \in \gst{S} : S \subseteq F \}$.
    A state $\gst{S}$ of $W(\mathcal{A})$ is called a \emph{game state}.
    It represents a situation where Alice can force the game to be in one of the sets $S \in \gst{S}$, and Bob can choose the actual state $q \in S$.
    The transition function $\delta_W $ is defined by
    \begin{align*}
      \delta_W (\{S\},A) & =
                           \{\{\delta(q,f(q)) \;|\; q \in S\} \;|\; f : S \rightarrow \{0,1\}\} \\
      \delta_W (\{S\},B) & =
                           \{\{\delta(q,b) \;|\; q \in S, b \in \{0,1\}\}\}
    \end{align*}
    and
    $\delta_W (\gst{S}, c) = \bigcup_{S \in \gst{S}} \delta_W (\{S\},c)$
    for a game state $\gst{S}$ and $c \in \{A, B\}$.

\end{definition}

The following observations follow easily from the definition of $W(L)$. 

\begin{lemma} \label{basicProperties}
    Let $\mathcal{A}$ be a DFA with alphabet $\{0,1\}$, and $W(\mathcal{A})$ the winning set DFA from Definition~\ref{winningSetAuto},
    and $\delta_W$ the iterated transition function for  $W(\mathcal{A})$.
    Let $\gst{P},\gst{R},\gst{S},\gst{T}$ be game states of $W(\mathcal{A})$, $P,R,S,T,V \subseteq 2^Q$ sets of states, and $w$ a word over $\{A,B\}$.
    \begin{roster}
        \item \label{unionEquiv} 
          If $\delta_W (\gst{P},w) = \gst{R}$ and $\delta_W (\gst{S},w) = \gst{T}$ then $\delta_W (\gst{P} \cup \gst{S}, w) = \gst{R} \cup \gst{T}$.
        \item \label{unionInsideEquiv}
          If $R \in \delta_W( \{S\},w )$ and $V \in \delta_W( \{T\},w )$,
          then some $P \in \delta_W( \{S \cup T\},w )$ satisfies $P \subseteq R \cup V$.
          Conversely, for each $P \in \delta_W(\{S \cup T\}, w)$ there exist $R \in \delta_W( \{S\},w )$ and $V \in \delta_W( \{T\},w )$ with $P = R \cup V$.
        \item \label{subsetEquiv} If $S,R \in \gst{T}$ and $S \subseteq R$, then $\gst{T} \sim \gst{T} \setminus \{R\}$.
        \item \label{removeStateEquiv} 
        If $S \in \gst{S}$ and some $q \in S$ has no path to a final state, then $\gst{S} \sim \gst{S} \setminus \{S\}$.
        \item \label{finalSinkEquiv}
        If $S \in \gst{S}$ and there is a sink state $q \in S \cap F$, then $\gst{S} \sim ( \gst{S} \setminus \{S\}) \cup \{S \setminus \{q\} \}$.
        \item \label{smallWordEquiv}
          If $\gst{S} = \gst{R} \cup \{S\}$ and the shortest path from some $q \in S$ to an final state in $\mathcal{A}$
        has length $\ell$, then for all $w \in \{A,B\}^{< \ell}$,
        $\delta_W(\gst{S},w)$ is final iff $\delta(\gst{R},w)$ is.
    \end{roster}
   
\end{lemma}

\begin{lemma} \label{lemmasEquivIncl}
  Recall the assumptions of Lemma~\ref{basicProperties}.
  \begin{roster}
  \item \label{equivChainDoubleIncl}
    Suppose that for every $S \in \gst{S}$ there exists $R \in \gst{R}$ with $R \subseteq S$, and reciprocally.
    Then $\gst{S} \sim \gst{R}$.
  \item  \label{equivChainSingletons}
    Let $v, w \in \{A, B\}^*$. 
    If for all $q \in Q$, the game states $\delta_W (\{\{q\}\},v)$ and $\delta_W (\{\{q\}\},w)$ are either both accepting or both rejecting, then $v \equiv_{W(\mathcal{L}(\mathcal{A}))} w$.
  \end{roster}
    
\end{lemma}

\begin{proof}
  \begin{roster}
  \item
    Let $w \in \{A,B\}^*$ be such that $\delta_W(\gst{S}, w)$ is accepting.
    Then some set $T \in \delta_W(\gst{S}, w)$ consists of accepting states of $\mathcal{A}$.
    By Lemma~\ref{basicProperties}\ref{unionEquiv} there exists $S \in \gst{S}$ with $T \in \delta_W(\{S\}, w)$.
    Let $R \in \gst{R}$ be such that $R \subseteq S$.
    Then there exists $V \in \delta_W(\gst{R}, w)$ with $V \subseteq T$ by Lemma~\ref{basicProperties}\ref{unionInsideEquiv}, so that $\delta_W(\gst{R}, w)$ is also accepting.
  \item
    Let $\gst{S} \in 2^{2^Q}$ be a game state and suppose $\delta_W(\gst{S}, v)$ is accepting, so there exists $P \in \delta_W(\gst{S}, v)$ consisting of accepting states of $\mathcal{A}$.
    By Lemma~\ref{basicProperties}\ref{unionEquiv} we have $\delta_W(\gst{S}, v) = \bigcup_{S \in \gst{S}} \delta_W(\{S\}, v)$, and similarly for $w$, so we may assume $\gst{S} = \{S\}$ is a singleton.
    By Lemma~\ref{basicProperties}\ref{unionInsideEquiv}, for each $q \in S$ there exists $R_q \in \delta_W(\{\{q\}\}, v)$ such that $P = \bigcup_{q \in S} R_q$.
    In particular each $R_q$ consists of accepting states of $\mathcal{A}$, so each $\delta_W(\{\{q\}\}, v)$ is accepting.
    Then $\delta_W(\{\{q\}\}, w)$ is also accepting, so there exists $T_q \in \delta_W(\{\{q\}\}, w)$ with $T_q \subseteq F$.
    By Lemma~\ref{basicProperties}\ref{unionInsideEquiv} there exists $P' \in \delta_W(\{S\}, w)$ with $P' \subseteq \bigcup_{q \in S} T_q \subseteq F$, and then $\delta_W(\{S\}, w)$ is accepting.
    This shows $v \equiv w$.
  \end{roster}
\end{proof}


\begin{proposition} \label{prop:upperBound}
    Let $\mathcal{A}$ an $n$-state DFA. 
    The number of states in the minimal DFA for $W(\mathcal{L}(\mathcal{A}))$ is at most the Dedekind number $D(n)$.
\end{proposition}
\begin{proof}
  The Dedekind number $D(n)$ is the number of antichains in $2^{2^Q}$ by inclusion, and every game state is equivalent to an antichain by Lemma 3.1(f).
  \end{proof}

  Note that the growth of $D(n)$ is doubly exponential in $n$.

  We have computed the exact state complexity of the winning set operation for DFAs with at most $5$ states; the $6$-state case is no longer feasible with our program and computational resources.
  The sequence begins with $1, 4, 16, 62, 517$.

\section{Doubly exponential lower bound}
    In this section we present the construction of a family of automata for which the number of states in the minimal
    winning set automaton is asymptotically optimal, that is to say doubly exponential.
    The idea is to reach any desired antichain of subsets of a special subset of states by reading the appropriate word, and then 
    to make sure these game states are nonequivalent by reading a word which leads to acceptance
    only if the game state is the wanted one (apart from some technical details).
    
    To do this we split the automaton into several components.
    First we present a ``subset factory gadget''
    that allows to make any
    desired 
    set of the form $\{ S\}$ where $S$ is a subset of a specific length-$n$ path in the automaton. This gadget will be used several times to accumulate subsets in the game state.
    Then we present a ``testing gadget'' allowing to distinguish between a doubly exponential number of game states.

    The construction of $W(\mathcal{A})$ in Definition~\ref{winningSetAuto} shows that the labels of the transitions
    are not important with regard to the winning set language that is obtained from it.
    In this section we define automata by describing their graphs, and a node with two outgoing transitions can have them arbitrary labeled by $0$ and $1$.

\begin{lemma}[Subset factory gadget]
    \label{lem:subsetFactoryGadget}
    Let $\GenSubset_n$ be the graph in Figure~\ref{subsetFactoryGadget}.
    For $i \in \{1, \ldots, n\}$, denote $o_i = e_{2n + i -1}$ (successors of the $c_i$).
    For all $S \in 2^{\{1,\ldots,n\}}$ there exists $w ^\mathrm{subset}_S \in \{A,B\}^{2n} $ such that
    $\delta_W ( \{\{b_1\}\}, w ^\mathrm{subset}_S) )  \sim  \{ \{ o_i \;|\; i \in S \}   \} $ for every DFA over $\{0,1\}$ that contains $\GenSubset_n$ as a subgraph.

    \begin{figure}[htp]
        \begin{center}
            \begin{tikzpicture}
              [nonaccept/.style={circle,draw,inner sep=0cm,minimum size=0.75cm},
              accept/.style={circle,draw,double,inner sep=0cm,minimum size=0.75cm},
	 every edge/.style={draw,->,>=stealth},
	 scale=0.42]
		\node [style=nonaccept] (0) at (-4, 2) {$b_1$};
		\node [style=nonaccept] (1) at (-1, 2) {$d_1$};
		\node [style=nonaccept] (2) at (5.75, 2) {$b_{n-1}$};
		\node [style=nonaccept] (3) at (9, 2) {$d_{n-1}$};
		\node [style=nonaccept] (4) at (-4, -1) {$c_1$};
		\node [style=nonaccept] (5) at (-6.25, -1) {$s_1$};
		\node [style=nonaccept] (6) at (3.25, -1) {$s_{n-1}$};
		\node [style=nonaccept] (7) at (5.75, -1) {$c_{n-1}$};
		\node [style=nonaccept] (8) at (11.25, -1) {$s_n$};
		\node [style=nonaccept] (9) at (13.5, 2) {$b_n$};
		\node [style=nonaccept] (10) at (13.5, -1) {$c_n$};
		\node [style=nonaccept] (11) at (13.5, -4) {$e_{3n - 1}$};
		\node [style=nonaccept] (12) at (5.75, -4) {};
		\node [style=nonaccept] (13) at (8.5, -4) {};
		\node [style=nonaccept] (14) at (11.25, -4) {};
		\node [style=nonaccept] (15) at (0.75, -4) {};
		\node [style=nonaccept] (16) at (-1.5, -4) {$e_2$};
		\node [style=nonaccept] (17) at (-4, -4) {$e_1$};
		\node  (18) at (3, 2) {$\phantom{I}$};
		\node  (19) at (0.75, 2) {$\phantom{I}$};
		\node  (20) at (2.5, -4) {$\phantom{I}$};
		\node  (21) at (4, -4) {$\phantom{I}$};
		\node  (22) at (2, 2) {$\cdots$};
		\node  (23) at (3.25, -4) {$\cdots$};
		\node  (24) at (-6, 2) {};
		\node [style=accept,inner sep=0cm,minimum size=0.75cm] (25) at (17.5, 2) {$b_{n+1}$};
		\node  (26) at (17.5, -4) {};
		\node  (27) at (17.5, -4) {$\phantom{I}$};
		
		\draw (0) edge (1);
		\draw (2) edge (3);
		\draw (0) edge (4);
		\draw (4) edge (5);
		\draw [in=130, out=-130, loop] (5) edge (5);
		\draw [in=145, out=-145, loop, looseness=5] (5) edge (5);
		\draw (7) edge (6);
		\draw [in=130, out=-130, loop] (6) edge (6);
		\draw [in=145, out=-145, loop, looseness=5] (6) edge (6);
		\draw (2) edge (7);
		\draw [out=10,in=170] (3) edge (9);
		\draw [out=-10,in=-170] (3) edge (9);
		\draw (9) edge (10);
		\draw (10) edge (8);
		\draw [in=130, out=-130, loop] (8) edge (8);
		\draw [in=145, out=-145, loop, looseness=5] (8) edge (8);
		\draw (7) edge (12);
		\draw [in=170, out=10] (12) edge (13);
		\draw [in=-170, out=-10] (12) edge (13);
		\draw (10) edge (11);
		\draw [in=170, out=10] (13) edge (14);
		\draw [in=-170, out=-10] (13)edge (14);
		\draw [in=170, out=10] (14) edge (11);
		\draw [in=-170, out=-10] (14) edge (11);
		\draw (4) edge (17);
		\draw [in=170, out=10] (17) edge (16);
		\draw [in=-170, out=-10] (17) edge (16);
		\draw [in=170, out=10] (16) edge (15);
		\draw [in=-170, out=-10] (16) edge (15);
		\draw [in=170, out=10] (1) edge (19);
		\draw [in=-170, out=-10] (1) edge (19);
		\draw [in=170, out=10] (18) edge (2);
		\draw [in=-170, out=-10] (18) edge (2);
		\draw [in=170, out=10] (15) edge (20);
		\draw [in=-170, out=-10] (15) edge (20);
		\draw [in=170, out=10] (21) edge (12);
		\draw [in=-170, out=-10] (21) edge (12);
		\draw (24) edge (0);
		\draw (9) edge (25);
		\draw [in=135, out=45, loop] (25) edge (25);
		\draw [in=120, out=60, loop, looseness=5] (25) edge (25);
		\draw [in=170, out=10] (11) edge (27);
		\draw [in=-170, out=-10] (11) edge (27);
\end{tikzpicture}
        \end{center}
        \caption{\label{subsetFactoryGadget} $\GenSubset_n$, the subset factory gadget.}
     \end{figure}

\end{lemma}

\begin{proof}
  Denote $f_i = e_{3 i - 2} $.
  For $i \in \{1, \ldots, n\}$ and $S \subseteq \{1, \ldots, i-1\}$, denote $S_i = \{e_{2 i - 4 + j} \;|\; j \in S\}$.
  Consider the game state $\gst{R}(i,S) = \{\{b_i\} \cup S_i\}$.
  If the automaton reads $A B$, the resulting game state is
  \[
    \delta_W(\gst{R}(i,S), A B) = \{\{s_i, e_{3 i}\} \cup S_{i+1}, \{b_{i+1}\} \cup S_{i+1}\} \sim \gst{R}(i+1, S) .
  \]
  In the case of $B A$ we instead have
  \[
    \delta_W(\gst{R}(i,S), B A) = \{ \{b_{i+1}, s_i\} \cup S_{i+1}, \{b_{i+1}, f_i\} \cup S_{i+1} \} \sim \gst{R}(i+1, S \cup \{i\}) .
  \]
  In both cases the final steps follow from Lemma~\ref{basicProperties}\ref{removeStateEquiv}.
  By Lemma~\ref{basicProperties}\ref{finalSinkEquiv} we also have $\gst{R}(n+1, S) \sim \{S_{n+1}\}$ since $b_{n+1}$ is an accepting sink state.
  
  Take $w^\mathrm{subset}_S$ as the concatenation $w_1 w_2 \dots w_n$ where $w_i = BA$ if $i \in S$, and $w_i = AB$ if $i \notin S$.
  This word satisfies the claim, since $\delta_W(\{\{b_1\}\}, w^\mathrm{subset}_S) = \{S_{n+1}\} = \{\{o_i \;|\; i \in S\}\}$.
\end{proof}

\begin{lemma}[Game state factory gadget]\label{subsetFactory}
    Let $\GenState_n$ be the graph in figure~\ref{stateFactoryGadget} and $\mathcal{A}$ any DFA over $\{0,1\}$ that contains it.
    For all $\gst{S} = \{ S_1, \ldots, S_\ell \}$ where each $S_i \subseteq \{r_1,\ldots, r_n\}$,
    there exists $w^\mathrm{gen}_{\gst{S}} \in \{A,B\}^{\ell (3 n + 1)}$, and a game state $\gst{S}'$ that does not contain a subset of the states of $\GenState_n$, such that
    $\delta_W (\{\{a_1\}\},w^\mathrm{gen}_{\gst{S}})  \sim \gst{S} \cup \{\{a_1\}\} \cup \gst{S}'$.

        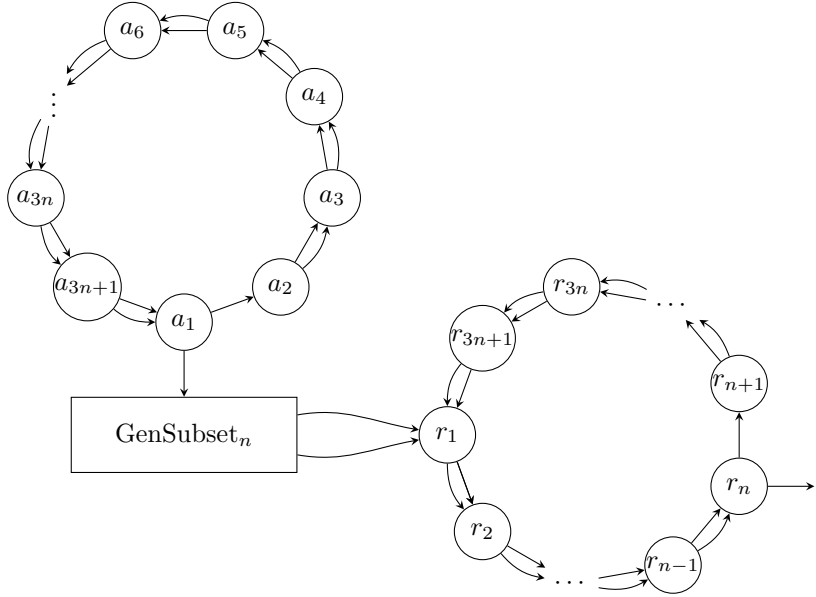
\begin{figure}[htp]
          \begin{center}
              \begin{tikzpicture}
	[nonaccept/.style={circle,draw,inner sep=0cm,minimum size=0.75cm},
	 accept/.style={circle,draw,double,inner sep=0cm,minimum size=0.75cm},
	 every edge/.style={draw,->,>=stealth},
	 scale=0.5]

	\node [nonaccept] (circ1) at (270:4cm) {$a_1$};
	\node [nonaccept] (circ2) at (310:4cm) {$a_2$};
	\node [nonaccept] (circ3) at (350:4cm) {$a_3$};
	\node [nonaccept] (circ4) at (30:4cm) {$a_4$};
	\node [nonaccept] (circ5) at (70:4cm) {$a_5$};
	\node [nonaccept] (circ6) at (110:4cm) {$a_6$};
	\node (circ7) at (150:4cm) {$\vdots$};
	\node [nonaccept] (circ8) at (190:4cm) {$a_{3n}$};
	\node [nonaccept] (circ9) at (230:4cm) {$a_{3n+1}$};

	\draw (circ1) edge (circ2);
	\foreach \i/\j/\ang in {2/3/40, 3/4/80, 4/5/120, 5/6/160, 6/7/200, 7/8/240, 8/9/280, 9/1/320}{
		 \draw (circ\i) edge (circ\j);
		 \draw [in=220+\ang, out=\ang] (circ\i) edge (circ\j);
	}

	\node [draw, rectangle, minimum height=1cm, minimum width=3cm] (gensubset) at (0,-7) {$\GenSubset_n$};
	\draw (circ1) edge (gensubset);

	\begin{scope}[shift={(11cm,-7cm)}]
	\node [nonaccept] (circ21) at (180:4cm) {$r_1$};
	\node [nonaccept] (circ22) at (220:4cm) {$r_2$};
	\node (circ23) at (260:4cm) {$\cdots$};
	\node [nonaccept] (circ24) at (300:4cm) {$r_{n-1}$};
	\node [nonaccept] (circ25) at (340:4cm) {$r_n$};
	\node [nonaccept] (circ26) at (20:4cm) {$r_{n+1}$};
	\node (circ27) at (60:4cm) {$\cdots$};
	\node [nonaccept] (circ28) at (100:4cm) {$r_{3n}$};
	\node [nonaccept] (circ29) at (140:4cm) {$r_{3n+1}$};
	\end{scope}

	\draw [in=170,out=10] (gensubset) edge (circ21);
	\draw [in=-170,out=-10] (gensubset) edge (circ21);
	
	\draw (circ21) edge (circ22);
	\foreach \i/\j/\ang in {1/2/0, 2/3/40, 3/4/80, 4/5/120, 6/7/200, 7/8/240, 8/9/280, 9/1/320}{
		 \draw (circ2\i) edge (circ2\j);
		 \draw [in=130+\ang, out=270+\ang] (circ2\i) edge (circ2\j);
	}
	\draw (circ25) edge (circ26);

	\node [right=1cm] (T) at (circ25) {};
	\draw (circ25) edge (T);

\end{tikzpicture}
        \end{center}
      
        \caption{\label{stateFactoryGadget} $\GenState_n$, the game state factory gadget.}
     \end{figure}

\end{lemma}

\begin{proof}
    The idea is that previously made subsets will rotate in the rightmost cycle.
    Meanwhile, a singleton set will rotate in the left cycle, initiating from the state $a_1$ the creation of a new subset by reading the letter $A$.
    This new set is created in the subset factory component and joins the previously made sets in the rightmost cycle.
    
    Suppose we have reached a game state of the form $\gst{R} = \{ \{a_1\},\{r_i \;|\; i \in S_1\} , 
    \ldots,  \{r_i \;|\; i \in S_k\}\} \cup \gst{S}'$ where $\gst{S}'$ does not contain any subset of $\GenState_n$.
    We prove that by reading $A w^\mathrm{subset}_{S_{k+1}} A^n$,
    we reach a game state of the form $\{ \{a_1\},\{r_i \;|\; i \in S_1\} , 
    \ldots,  \{r_i \;|\; i \in S_{k+1}\}\} \cup \gst{S}''$.
    We analyze the elements of $\gst{R}$ separately.
    \begin{itemize}
        \item Because $|A w^{subset}_{S_{k+1}}A^{n}| = 3n+1$
        is the size of the rightmost cycle, we have $\delta_W (\{ \{r_i \;|\; i \in S_j\} \}, A w^{subset}_{S_{k+1}} A^{n}) \sim \{ \{r_i : i \in S_j\}\}$ for each $j \leq k$.
            
      \item The game state $\{\{a_1\}\}$ first evolves into $\delta_W(\{\{a_1\}\}, A) = \{ \{a_2\},\{b_1\}\}$.
        The component $\{\{a_2\}\}$ becomes $\{\{a_1\}\}$ when we read $w^{subset}_{S_{k+1}}A^{n}$.
        As for $\{\{b_1\}\}$, Lemma~\ref{subsetFactory} gives $\delta_W(\{\{b_1\}\}, w^{subset}_{S_{k+1}} )  = 
        \{ \{ o_i \;|\; i\in S_{k+1}\}\}$,
       and then
       $\delta_W(\{ \{ o_i \;|\; i\in S_{k+1}\} \}, A^n ) = 
       \{ \{ r_{i} \;|\; i \in S_{k+1}\} \} \cup \gst{S}''$
       where every set in $\gst{S}''$ contains a state outside of $\GenState_n$.
     \item
       The game state $\gst{S}'$ evolves into some $\gst{S}'''$ each of whose sets contains a state not in $\GenState_n$, since the gadget cannot be re-entered.
    \end{itemize}
    
    By Lemma~\ref{basicProperties}\ref{unionEquiv} we have 
    $\delta_W (\gst{R}, A w^\mathrm{subset}_{S_{k+1}} A^{n})
    \sim \{ \{a_1\},\{r_{i} \;|\; i \in S_1\}, 
    \ldots,  \{r_{i} \;|\; i \in S_{k+1}\}\} \cup \gst{S}'' \cup \gst{S}'''$.
    We obtain $w^\mathrm{gen}_{\gst{S}}$ as the concatenation of these words.
  \end{proof}

\begin{lemma}[Testing gadget] \label{lem:testingGadget}
    Let $\Testing_n$ be the graph in Figure~\ref{measuringGadget}.
    \begin{roster}
        \item For all $P \subseteq \{1, \dots, n\}$
        there exists $w ^\mathrm{test}_P \in \{A,B\}^{n} $ such that
        for each $I \subseteq \{1, \dots, n\}$, the game state $\delta_W (  \{\{ q_i \;|\; i \in I\}\}, w^\mathrm{test}_P)$ is accepting iff $I \subseteq P$.
        \item Let $V$ be the set of nodes of the graph $\Testing_n$.
                Then for all $\gst{S} \in 2^{2^{V}}$ and $w\in \{A, B\}^{\geq 2n}$, the game state $\delta_W (\gst{S},w)$ is not accepting.
    \end{roster}
    
    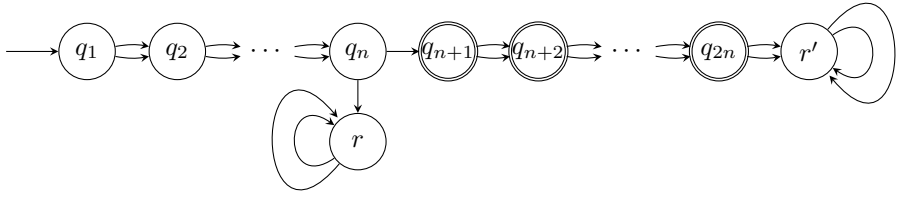
\begin{figure}[htp]
       \begin{center}
        \begin{tikzpicture}
	[nonaccept/.style={circle,draw,inner sep=0cm,minimum size=0.75cm},
	 accept/.style={circle,draw,double,inner sep=0cm,minimum size=0.75cm},
	 every edge/.style={draw,->,>=stealth},
	 scale=0.4]

	 \node (x) at (-3,0) {};
	 \node [nonaccept] (q1) at (0,0) {$q_1$};
 	 \node [nonaccept] (q2) at (3,0) {$q_2$};
 	 \node (q3) at (6,0) {$\cdots$};
 	 \node [nonaccept] (q4) at (9,0) {$q_n$};

	 \node [nonaccept] (r) at (9,-3) {$r$};

	 \node [accept] (a1) at (12,0) {$q_{n+1}$};
	 \node [accept] (a2) at (15,0) {$q_{n+2}$};
	 \node (a3) at (18,0) {$\cdots$};
	 \node [accept] (a4) at (21,0) {$q_{2 n}$};

	 \node [nonaccept] (r2) at (24,0) {$r'$};

	 \draw (x) edge (q1);

	 \foreach \i/\j in {q1/q2,q2/q3,q3/q4,a1/a2,a2/a3,a3/a4,a4/r2}{
	 	  \draw [in=170, out=10] (\i) edge (\j);
		  \draw [in=-170, out=-10] (\i) edge (\j);
	 }
	 \draw (q4) edge (r);
	 \draw (q4) edge (a1);

	\draw [in=130, out=-130, loop] (r) edge (r);
	\draw [in=145, out=-145, loop, looseness=5] (r) edge (r);

	\draw [in=-50, out=50, loop] (r2) edge (r2);
	\draw [in=-35, out=35, loop, looseness=5] (r2) edge (r2);

\end{tikzpicture}
       \end{center}

        \caption{\label{measuringGadget} $\Testing_n$, the testing gadget.}
     \end{figure}

   \end{lemma}

\begin{proof}
    \begin{roster}
        \item
          For $I \subseteq \{1, \ldots, 2 n\}$, denote $S_I = \{ q_i \;|\; i \in I \}$.
          If $2 n \notin I$, let $J = \{ i+1 \;|\; i \in I \}$.
          A simple case analysis together with Lemma~\ref{basicProperties}\ref{removeStateEquiv} and~\ref{finalSinkEquiv} shows that $\delta_W(\{S_I\}, A) \sim \{S_J\}$ and
          \[
            \delta_W(\{S_I\}, B) =
            \begin{cases}
              \emptyset & \mbox{if } n \in I, \\
              \{S_J\} & \mbox{otherwise.}
            \end{cases}
          \]
          Take
          \[
            w^\mathrm{test}_P[i] =
            \begin{cases}
              A &\mbox{if } n-i+1 \in P, \\
              B &\mbox{otherwise. }
            \end{cases}
          \]
          Then $\gst{S} = \delta_W (\{S_I\}, w^\mathrm{test}_P)$ is accepting if and only if $\gst{S} \sim \{\{q_{i+n} \;|\; i \in I\}\}$.
          This is equivalent to $w[n-i+1] = A$ for all $i \in I$, i.e. $I \subseteq P$.
        \item If $\gst{S} \in 2^{2^{V}}$ and $w$ with $|w| \geq 2n$,
          then every $S \in \delta_W (\gst{S},w)$ satisfies $S \subseteq \{r,r'\}$.
        \end{roster}
  \end{proof}

\begin{theorem}
    For each $n > 0$ there exists a DFA
    $\mathcal{A}_n$ over $\{0,1\}$ with $15n + 3$ states such that
    the minimal DFA for $W(\mathcal{L}(\mathcal{A}_n))$ has a least $D(n)$ states.
  \end{theorem}

  Together with Proposition~\ref{prop:upperBound}, this implies that the state complexity of $W$ restricted to regular languages grows doubly exponentially.
  
\begin{proof}
  Let $\mathcal{A}_n$ be the DFA obtained by combining $\Testing_n$ with the outgoing arrow of $\GenState_n$ and assigning $a_1$ as the initial state.
  
  For an antichain $\gst{S}$ on the powerset of $\{r_1, \ldots, r_n\}$, let $X_\gst{S} = \delta_W (\{\{a_1\}\}, w^\mathrm{gen}_\gst{S})$.
  By Lemma~\ref{subsetFactory} we have $X_\gst{S} \sim \{\{a_1\}\} \cup \gst{S} \cup \gst{S}'$ where each set in $\gst{S}'$ contains a state of $\Testing_n$.
  By definition, $X_\gst{S}$ is an accessible state of $W(\mathcal{A})$, and we show that distinct antichains $\gst{S}$ result in nonequivalent states.
  
  Let $P \subseteq \{1, \ldots, n\}$ and consider the game state $X'_\gst{S} = \delta_W(X_\gst{S}, A^{n+1} w^\mathrm{test}_P)$.
  We claim that $X'_\gst{S}$ is accepting iff some element of $\gst{S}$ is a subset of $\{ r_i \;|\; i \in P \}$.
  By Lemma~\ref{basicProperties}\ref{unionEquiv} we may analyze the components of $X_\gst{S}$ separately.
  \begin{itemize}
  \item
    Since the shortest path from $a_1$ to an accepting state has lenght $2 n + 2$ and $|A^{n+1} w^\mathrm{test}_P| = 2 n + 1$, we can ignore it by Lemma~\ref{basicProperties}\ref{smallWordEquiv}.
  \item
    Since each set of $\gst{S}'$ contains a state of $\Testing_n$ and $|A^{n+1} w^\mathrm{test}_P| \geq 2 n$, by Lemma~\ref{lem:testingGadget} the game state $\delta_W(\gst{S}', A^{n+1} w^\mathrm{test}_P)$ is not accepting.
  \item
    The game state $\delta_W(\gst{S}, A^{n+1})$ consists of the sets $\{ q_i \;|\; r_i \in S \}$ for $S \in \gst{S}$, as well as sets that contain at least one element of $\{r_{n+1}, \ldots, r_{2 n}\}$.
    We can ignore the latter by Lemma~\ref{basicProperties}\ref{smallWordEquiv}.
    Lemma~\ref{lem:testingGadget} shows that the former sets produce an accepting game state in $X'_\gst{S}$ iff some $S \in \gst{S}$ is a subset of $\{ r_i \;|\; i \in P \}$.
  \end{itemize}
  The Dedekind number $D(n)$ is the number of antichains on the powerset of $\{1, \ldots, n\}$, so we have found $D(n)$ nonequivalent states in $W(\mathcal{A})$.
\end{proof}


\section{Case of the bounded regular languages}

In this section we prove an upper bound on the complexity of the winning set of a bounded regular language.
Our proof technique is based on tracing the evolution of individual states of a DFA $\mathcal{A}$ in the winning set automaton $W(\mathcal{A})$ when reading several $A$-symbols in a row.

\begin{definition}[Histories of Game States]
  \label{def:History}
  Let $\mathcal{A} = (Q, \{0,1\}, q_0, \delta, F)$ be a DFA.
  Let $\gst{S} \in 2^{2^Q}$ be a game state of $W(\mathcal{A})$, and for each $i \geq 0$, let $\gst{S}_i \sim \delta_W(\gst{S}, A^i)$ be the game state with all supersets removed as per Lemma~\ref{basicProperties}\ref{subsetEquiv}.
  A \emph{history function} for $\gst{S}$ is a function $h$ that associates to each $i > 0$ and each set $S \in \gst{S}_i$ a \emph{parent set} $h(i, S) \in \gst{S}_{i-1}$, and to each state $q \in S$ a set of \emph{parent states} $h(i, S, q) \subseteq h(i, S)$ such that
  \begin{itemize}
  \item $S \in \delta_W(\{h(i,S)\}, A)$ for each $i$,
  \item $h(i,S)$ is the disjoint union of $h(i, S, q)$ for $q \in S$, and
  \item $\{q\} \in \delta_W(\{h(i,S,q)\}, A)$ for all $q \in S$.
  \end{itemize}
  
  The \emph{history} of a set $S \in \gst{S}_i$ from $i$ under $h$ is the sequence $S_0, S_1, \ldots, S_i = S$ with $S_{j-1} = h(j, S_j)$ for all $0 < j \leq i$.
  A history of a state $q \in S$ in $S$ under $h$ is a sequence $q_0, \ldots, q_i = q$ with $q_{j-1} \in h(j, S_j, q_j)$ for all $0 < j \leq i$.
\end{definition}

A game state can have several different history functions, and each of them defines a history for each set $S$.
A state of $S$ can have several histories under a single history function.
These histories are consistent with themselves and each other.
The proof of the main result of this section is based on the idea of choosing a ``good'' history function.
Note that we have defined the history function only for sequences of $A$-symbols, since this simplifies the definition and histories with $B$-symbols are not used in the proof.

For the rest of this section, we fix an $n$-state DFA $\mathcal{A} = (Q, \{0,1\}, q_0, \delta, F)$ that recognizes a bounded binary language and has disjoint cycles.
Let the lengths of the cycles be $k_1, \ldots, k_p$, and let $\ell$ be the number of states not part of any cycle.

We define a preorder ${\leq}$ on the state set $Q$ by reachability: $p \leq q$ holds if and only if there is a path from $p$ to $q$ in $\mathcal{A}$.
The notation $p < q$ means $p \leq q$ and $q \not\leq p$.
For two history functions $h, h'$ of a game state $\gst{S}$, we write $h \leq h'$ if for each $i > 0$, each $S \in \gst{S}_i$ and each $q \in S$, there exists a function $f : h(i, S, q) \to h'(i, S, q)$ with $p \leq f(p)$ for all $p \in h(i, S, q)$.
This defines a preorder on the set of history functions of $\gst{S}$.
We write $h < h'$ if $h \leq h'$ and $h' \not\leq h$.
A history function $h$ is \emph{minimal} if there exists no history function $h'$ with $h' < h$.
Intuitively, a minimal history function is one where the histories of states stay in the early cycles of $\mathcal{A}$ as long as possible.

\begin{lemma}
  \label{existsLocallyMinimal}
  Each game state $\gst{S} \in 2^{2^Q}$ has at least one minimal history function.
\end{lemma}

\begin{proof}
For each $i > 0$ and $S \in \gst{S}_i$, the set of possible choices for the parent $h(i, S)$ of $S$ and the parent set $h(i, S, q)$ of each state $q \in S$ is finite, and the choice is independent of the respective choices for other sets $S' \in \gst{S}_{i'}$ with $S' \neq S$ or $i' \neq i$.
If we choose the parents that are minimal with respect to ${\leq}$ for each set, the resulting history function is minimal.
\end{proof}

\begin{lemma}
\label{lem:ChangeHistory}
Let $\gst{S} \in 2^{2^Q}$ be any game state of $W(\mathcal{A})$.
Then there exist $k \leq \lcm(k_1, \ldots, k_p) + 2 n + \max_{x \neq y} \lcm (k_x, k_y)$ and $m \leq \lcm(k_1, \ldots, k_p)$ such that $\delta_W(\gst{S}, A^k) \sim \delta_W(\gst{S}, A^{k+m})$.
\end{lemma}

\begin{proof}
Denote the cycles of $\mathcal{A}$ by $C_1, \ldots, C_p$, so that $|C_i| = k_i$ for each $i$.
Let $h$ be a minimal history function of $\gst{S}$, given by Lemma~\ref{existsLocallyMinimal}.
Define $\gst{S}_i$ for $i \geq 0$ as in Definition~\ref{def:History}.

Let $t \geq 0$, $S \in \gst{S}_t$ and $q \in S$ be arbitrary, and let $S_0, \ldots, S_t = S$ and $q_0, \ldots, q_t = q$ be their histories under $h$.
The history of $q$ travels through some of the cycles of $\mathcal{A}$, never entering the same cycle twice.
We split the sequence $q_0, \ldots, q_n$ into words over $Q$ as $u_0 v_1^{p_1} u_1 v_2^{p_2} u_2 \cdots v_r^{p_r} u_r$, where
\begin{itemize}
\item
  each $p_j \geq 1$,
\item
  each $v_j$ consists of the states of some cycle, which we may assume is $C_j$, repeated exactly once,
\item
  the $u_j$ do not repeat states and each $u_j$ does not contain any states from $C_{j+1}$.
\end{itemize}
Intuitively, $v_j$ represents a phase of the history where the state stays in a cycle for several loops, and the $u_j$ represent transitions from one loop to another.
Each $u_j$ ends right before the time step when the history of $q$ enters the loop $C_{j+1}$.
It may share a nonempty prefix with $v_j$.

We claim that $p_i k_i \leq \max_{x \neq y} \lcm (k_x, k_y)$ holds for all $1 < i \leq r$.
Assume the contrary.
Since $k_i = |v_i|$, we have in particular $p_i |v_i| > \lcm(|v_{i-1}|, |v_i|)$ for some $i$, so that $a |v_{i-1}| = b |v_i|$ holds for some $a > 0$ and $0 < b \leq p_i$.
Denote $s = |u_0 v_1^{p_1} \cdots u_{i-1}|$, which is the time step after which the history of $q$ enters the repetitive portion of the previous loop $C_{i-1}$.
Denote $K = |v_{i-1}^a u_{i-1}|$ and $q'_{s+1}, q'_{s+2}, \ldots, q'_{s + K} = v_{i-1}^a u_{i-1}$.
Note that we may have $q'_{s+j} = q_{s+j}$ for some $1 \leq j < K$, but $q'_{s+K} < q_{s+K}$ since the former is not in $C_i$ while the latter is.
In $\mathcal{A}$ we have transitions from $q_s$ to both $q_{s+1}$ and $q'_{s+1}$, and from each $q'_{s+j}$ to $q'_{s+j+1}$, as well as from $q'_{s+K}$ to $q_{s+K+1}$.
For $q \leq j \leq K$ the game state $\delta_W(\gst{S}, A^{s+j})$ contains $S'_{s+j} =: (S_{s+j} \setminus \{q_{s+j}\}) \cup \{q'_{s+j}\}$.
We also have $S_{s+K+1} \in \delta_W(\{S'_{s+K}\}, A)$.
There are now two possibilities.
If $q'_{s+K} \in S_{s+K}$, then $S'_{s+K} \in \delta_W(\gst{S}, A^{s+K})$ is a proper subset of $S_{s+K}$, which contradicts our choice of $\gst{S}_{s+K}$ as a version of $\delta_W(\gst{S}, A^{s+K})$ with all proper supersets removed.
If $q'_{s+K} \notin S_{s+K}$, then we may define a new history function $h'$ by defining $h'(s+K+1, S) = S'_{a+K}$, $h'(s+K+1, S, q) = (h(s+K+1, S, q) \setminus \{q_{s+K}\}) \cup \{q'_{s+K}\}$, and $h'(t, S', q') = h(t, S', q')$ for all other choices of $t$, $S'$ and $q'$.
Then the function $f : h'(s+K+1, S, q) \to h(s+K+1, S, q)$ defined by $f(q'_{s+K}) = q_{s+K}$ and $f(q') = q'$ for other $q' \in h'(s+K+1, S, q)$ shows $h' < h$, which contradicts the local minimality of $h$.

We have now shown $p_i k_i \leq \max_{x \neq y} \lcm (k_x, k_y)$ for all $1 < i \leq r$.
Denote $L = \lcm(k_1, \ldots, k_p)$.
If $t \geq L + 2 \ell + p \cdot \max_{x \neq y} \lcm (k_x, k_y)$, then $p_1 k_1 \geq L + \ell$, which implies $q_\ell = q_{\ell+L}$ (note that $|u_0| \leq \ell$, so that $q_\ell, q_{\ell + L} \in C_1$).
Since this holds for every history of every state of $S$ under $h$ and each state of each set $S_i$ for $i \leq t$ can be chosen as $q_i$ for some $q \in S$, we have $S_\ell = S_{\ell+L}$.
Then $S \in \delta_W(\{S_{\ell+L}\}, A^{t-\ell-L}) = \delta_W(\{S_\ell\}, A^{t-\ell-L})$, so in particular $S \in \delta_W(\gst{S}, A^{t-L}) \sim \gst{S}_{t-L}$.
On the other hand, $S \in \delta_W(\{S_\ell\}, A^{t-\ell}) = \delta_W(\{S_{\ell+L}\}, A^{t-\ell})$, so $S \in \delta_W(\gst{S}, A^{t+L}) \sim \gst{S}_{t+L}$.
Since $S \in \gst{S}_t$ was arbitrary, we have $\gst{S}_t \sim \gst{S}' \subseteq \gst{S}_{t-L}$ and $\gst{S}_t \sim \gst{S}'' \subseteq \gst{S}_{t+L}$ for some game states $\gst{S}', \gst{S}'' \in 2^{2^Q}$.
By considering $t+L$ instead of $t$ and doing the same analysis, we obtain $\gst{S}_t \sim \gst{S}_{t+L}$.
\end{proof}

\begin{theorem}
Let $\mathcal{A}$ be an $n$-state DFA that recognizes a bounded binary laguage.
Then there is a partition $\ell + k_1 + \cdots + k_p = n$ such that the minimal DFA for $W(\mathcal{L}(\mathcal{A}))$ has at most $\sum_{m = 0}^{\ell + p + 1} (p \cdot \max_{x \neq y} \lcm(k_x, k_y) + 2 \ell + 2 \lcm(k_1, \ldots, k_p))^m$ states.
\end{theorem}

\begin{proof}
Denote the minimal DFA for $W(\mathcal{L}(\mathcal{A}))$ by $\mathcal{B}$.
We may assume that $\mathcal{A}$ is minimal, and then it has disjoint cycles.
Let $k_1, \ldots, k_p$ be the lengths of the cycles and $\ell$ the number of remaining states, and denote $P = p \cdot \max_{x \neq y} \lcm(k_x, k_y) + 2 \ell + 2 \lcm(k_1, \ldots, k_p)$.
Then the language of $W(\mathcal{A})$ only contains words that have at most $\ell + p$ occurrences of $B$: in a game whose turn order has more $B$s than that, Bob can win by choosing to leave a cycle whenever possible, since the $\ell$ states outside the cycles can never be returned to.

Consider a word $w = A^{t_0} B A^{t_1} B \cdots B A^{t_m}$ with $0 \leq m \leq \ell + p$.
If $t_i \geq P$ for some $i$, then Lemma~\ref{lem:ChangeHistory} implies $\delta_W(\gst{S}, A^{t_i}) \sim \delta_W(\gst{S}, A^t)$ for the game state $\gst{S} = \delta_W(\{\{q_0\}\}, A^{t_0} B \cdots A^{t_{i-1}} B)$ and some $t < t_i$.
Thus the number of distinct states of $\mathcal{B}$ reachable by words of this form is at most $P^{m+1}$.
The claim directly follows.
\end{proof}

The state complexity implied by the result (the maximum of the expression taken over all partitions of $n$) is at least $n^n$.
In particular, it grows superexponentially.
We don't know whether the actual complexity of the winning set operation on bounded regular languages is exponential or not.
If we combine the gadgets $\GenSubset_n$ and $\Testing_n$, the resulting DFA recognizes a language whose winning set requires at least $2^n$ states,
so for finite regular languages (and therefore for bounded regular languages) the state complexity of the winning set is at least exponential.

\section{Chain automata}

In this section we investigate a family of binary automata consisting of a chain of states with a self-loop on each state.
More formally, a \emph{chain automaton} is a DFA $\mathcal{A} = (Q, \{0,1\}, q_0, \delta, F)$ where $Q = \{0, 1, \ldots, m+p-1\}$, $q_0 = 0$, $\delta(i, 0) = i$ and $\delta(i, 1) = i+1$ for all $i \in Q$ except $\delta(m+p-1, 1) = m$.
The automaton is \emph{$1$-bounded} if $p = 0$ and the state $m-1$ is not final.
See Figure~\ref{chainGeneral} for a $1$-bounded chain automaton.
It is easy to see that chain automata recognize exactly the regular languages $L$ such that $w \in L$ depends only on $|w|_1$, and the $1$-bounded subclass recognizes those where $|w|_1$ is also bounded.
Of course, the labels of the transitions have no effect on the winning set $W(\mathcal{L}(\mathcal{A}))$ so the results of this section apply to every DFA with the structure of a chain automaton.

\begin{figure}[htp]
    \begin{center}
        \begin{tikzpicture}
	[nonaccept/.style={circle,draw,inner sep=0cm,minimum size=0.9cm},
	 accept/.style={circle,draw,double,inner sep=0cm,minimum size=0.9cm},
	 every edge/.style={draw,->,>=stealth},
	 scale=0.43]
		\node (i) at (-10, 0) {};
		\node [nonaccept] (n0) at (-7, 0) {$0$};
		\node [nonaccept] (n1) at (-3, 0) {$1$};
		\node [accept] (n2) at (1, 0) {$2$};
		\node [nonaccept] (n3) at (5, 0) {$3$};
		\node (n4) at (8, 0) {$\cdots$};
		\node [accept] (n5) at (11, 0) {$n-2$};
		\node [nonaccept] (n6) at (15, 0) {$n-1$};

		\draw (i) edge (n0);
		\foreach \k/\kk in {0/1,1/2,2/3,3/4,4/5,5/6}{
			 \draw (n\k) edge node [midway,above] {$0$} (n\kk);
		}
		\foreach \k in {0,1,2,3,5,6}{
			\draw [in=55, out=125, loop, looseness=4] (n\k) edge node [midway,above] {$1$} (n\k);
		}
		\draw [in=40, out=140, loop] (n6) edge node [midway,above] {$0$} (n6);
		
\end{tikzpicture}

    \end{center}

    \caption{\label{chainGeneral} A $1$-bounded chain automaton. Note that the state $n-1$ is not final.}
\end{figure}
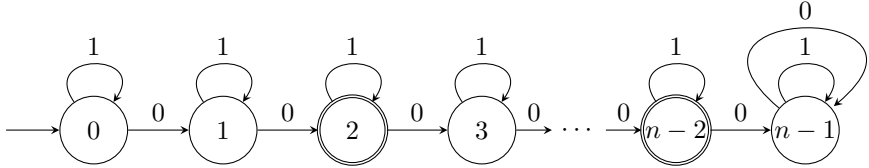

\begin{lemma}
  \label{chainEquivWords}
  Let $\mathcal{A}$ be an $n$-state chain automaton, and denote $\equiv_{W(\mathcal{L}(\mathcal{A}))}$ by $\equiv$.
  \begin{roster}
  \item
    \label{chain1} For every state $q \in Q$ and every $S \in \delta_W(\{\{q\}\}, A B)$, there exists $R \in \delta_W(\{\{q\}\}, B A)$ with $R \subseteq S$.
  \item
    \label{chain2} For all $k \in \mathbb{N}$, $B^k A^k B^{k+1} \equiv B^{k+1} A^k B^k$. 
  \item
    \label{chain3} For all $k \in \mathbb{N}$, $A^{k+1} B^k A^k \equiv A^k B^k A^{k+1}$. 
  \item
    \label{chain4} $A^{n-1} \equiv A^n$ and $B^{n-1} \equiv B^n$.
  \end{roster}
\end{lemma}

The intuition for~\ref{chain1} is that $B A$ produces game states that are better for Alice than $A B$, since Alice can undo any damage Bob just caused.

\begin{proof}
  Label the states of $\mathcal{A}$ by $\{0, 1, \ldots, m+p-1\}$ as in the definition of chain automata.
  For \ref{chain1}, \ref{chain2} and \ref{chain3} we ``unroll'' the loop $m, m+1, \ldots, m+p-1$ to obtain
  an equivalent automaton with an infinite chain of states, simplifying the arguments.
  We also argue in terms of the NFA for $W(\mathcal{L}(\mathcal{A}))$, saying that ``we produce a set $R \subseteq Q$ from $S \subseteq Q$ by reading $w \in \{A,B\}^*$'' if $R \in \delta_W(\{S\}, w)$.
  
  In this formalism \ref{chain1} means that for any set produced from $\{q\}$ by reading $A B$, we can produce a subset of it by reading $B A$.
  To see what sets can be produced, we use
  a spacetime diagram (Figure~\ref{ABinclBA}) where the time increases to the south. In this 
  diagram, by reading a $B$, a selected state will spread south and southeast. By reading an
  $A$, we have both possibilities, resulting in multiple sets.
  The claim follows directly.

    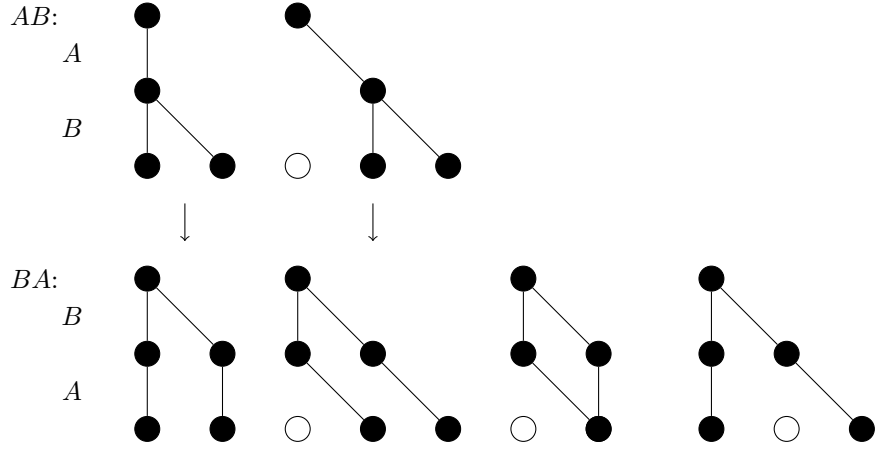
\begin{figure}[htp]
    
        \begin{center}
            \begin{tikzpicture}[scale=0.5]
		\node  (0) at (-4, 4) {$AB$:};
		\node  (1) at (-4, -3) {$BA$:};
		\node [style=dot] (2) at (-1, 4) {};
		\node [style=dot] (3) at (-1, 2) {};
		\node [style=dot] (4) at (-1, 0) {};
		\node [style=dot] (5) at (1, 0) {};
		\node [style=dot] (6) at (3, 4) {};
		\node [style=dot] (7) at (5, 2) {};
		\node [style=dot] (8) at (5, 0) {};
		\node [style=dot] (9) at (7, 0) {};
		\node [style=dot] (10) at (-1, -3) {};
		\node [style=dot] (11) at (-1, -5) {};
		\node [style=dot] (12) at (-1, -7) {};
		\node [style=dot] (13) at (1, -7) {};
		\node [style=dot] (14) at (14, -3) {};
		\node [style=dot] (15) at (16, -5) {};
		\node [style=dot] (17) at (18, -7) {};
		\node [style=dot] (18) at (1, -5) {};
		\node [style=dot] (19) at (14, -5) {};
		\node [style=dot] (20) at (14, -7) {};
		\node [style=dot] (21) at (3, -3) {};
		\node [style=dot] (22) at (3, -5) {};
		\node [style=dot] (23) at (5, -7) {};
		\node [style=dot] (24) at (7, -7) {};
		\node [style=dot] (25) at (5, -5) {};
		\node [style=dot] (26) at (9, -3) {};
		\node [style=dot] (27) at (11, -5) {};
		\node [style=dot] (28) at (11, -7) {};
		\node [style=dot] (29) at (9, -5) {};
		\node [style=dot] (30) at (11, -7) {};
		\node  (35) at (0, -2) {};
		\node  (36) at (0, -1) {};
		\node  (37) at (5, -1) {};
		\node  (38) at (5, -2) {};
		\node  (39) at (-3, 3) {$A$};
		\node  (40) at (-3, 1) {$B$};
		\node  (41) at (-3, -4) {$B$};
		\node  (42) at (-3, -6) {$A$};
		\node [style=small] (43) at (16, -7) {};
		\node [style=small] (44) at (3, -7) {};
		\node [style=small] (45) at (3, 0) {};
		\node [style=small] (46) at (9, -7) {};
		\draw (2) to (3);
		\draw (3) to (5);
		\draw (3) to (4);
		\draw (6) to (7);
		\draw (7) to (8);
		\draw (7) to (9);
		\draw (10) to (11);
		\draw (11) to (12);
		\draw (14) to (15);
		\draw (15) to (17);
		\draw (10) to (18);
		\draw (18) to (13);
		\draw (14) to (19);
		\draw (19) to (20);
		\draw (21) to (22);
		\draw (22) to (23);
		\draw (21) to (25);
		\draw (25) to (24);
		\draw (26) to (27);
		\draw (27) to (28);
		\draw (26) to (29);
		\draw (29) to (30);
		\draw [style=new edge style 0] (36.center) to (35.center);
		\draw [style=new edge style 0] (37.center) to (38.center);
\end{tikzpicture}

        \end{center}

        \caption{\label{ABinclBA} Sets produced with $AB$ and $BA$. The arrows indicate a subset relation.}
    \end{figure}

    We now prove \ref{chain2}.
    By Lemma~\ref{lemmasEquivIncl}\ref{equivChainSingletons} it is enough to consider singleton states $\{q\}$, and without loss of generality we assume $q = 0$.
    Because of (a) it is sufficient to prove that for every set obtained
    from $B^{k+1} A^k B^k$, we can produce a subset of it by reading $B^k A^k B^{k+1}$.
    After reading $B^{k+1}$, we have the interval $\{0, \dots k+1\}$. After $A^k$ 
    we get a set included in $\{0,\dots, 2k+2\}$ where the distance between every two consecutive elements is less than $k$.
    After reading $B^k$ the gaps are filled and we get an interval containing
    $ \{k, \dots, 2k+2\}$.
    When reading $B^k A^k B^{k+1}$ we do the following: get $\{0, \dots k\}$ with $B^k$,
    make every element go to position $k$ with $A^k$, and extend with $B^{k+1}$ to get the interval $\{k, \dots, 2k+2\}$.

    For \ref{chain3}, for the same reason as previously it is sufficient to prove that for every set obtained from $\{0\}$
    by reading $A^k B^k A^{k+1}$, we can produce a subset by reading $A^{k+1} B^k A^k$.
    First we prove that by reading $A^{k+1} B^k A^k$ we can get any singleton set
    $\{k\}, \{k+1\}, \ldots, \{2k+1\}$: By reading $A^{k+1} B^k$, we can get any singleton
    set between $0$ and $k+1$ and expand it to have any interval of length $k+1$
    between $0$ and $2k+1$.
    Then by reading $A^k$ we can have a singleton state at the end position
    of the interval, that is between $k$ and $2k+1$.
    See Figure~\ref{AnABnAn}.

    \begin{figure}[htp]
    
        \begin{center}
            \begin{tikzpicture}[scale=0.6]
		\node [fill,circle] (48) at (0, 13) {};
		\node [draw,dotted,circle,thick] (49) at (0, 8) {};
		\node [draw,dotted,circle,thick] (50) at (6, 8) {};
		\node [draw,dotted,circle,thick] (51) at (1, 8) {};
		\node [draw,dotted,circle,thick] (52) at (2, 8) {};
		\node [draw,dotted,circle,thick] (53) at (3, 8) {};
		\node [draw,dotted,circle,thick] (54) at (4, 8) {};
		\node [draw,dotted,circle,thick] (55) at (5, 8) {};
		\node  (56) at (0, 9) {};
		\node  (57) at (2, 10.5) {};
		\node  (58) at (4.5, 9) {};
		\node  (59) at (6, 8.75) {$k+1$};
		\node  (60) at (0, 8.5) {$0$};
		\node  (61) at (1, 5) {};
		\node  (62) at (6, 5) {};
		\node  (63) at (1, 8) {};
		\node  (64) at (0, 5) {};
		\node  (65) at (12, 5) {};
		\node  (66) at (12, 6) {$2k + 1$};
		\node  (67) at (6, 2) {};
		\node  (68) at (-2, 11) {$A^{k+1}$};
		\node  (69) at (-2, 6.5) {$B^k$};
		\node  (70) at (-2, 3.25) {$A^k$};
		\node  (73) at (4, 2.5) {$k$};
		\node  (74) at (12, 2.75) {$2k + 1$};
		\node  (75) at (1, 10.5) {};
		\node  (76) at (3.5, 9) {};
		\node  (77) at (3.5, 10.5) {};
		\node  (78) at (6, 9) {};
		\node  (79) at (3, 5.5) {length $k+1$};
		\node [draw,dotted,circle,thick] (80) at (5, 2) {};
		\node [draw,dotted,circle,thick] (81) at (11, 2) {};
		\node [draw,dotted,circle,thick] (82) at (6, 2) {};
		\node [draw,dotted,circle,thick] (83) at (7, 2) {};
		\node [draw,dotted,circle,thick] (84) at (8, 2) {};
		\node [draw,dotted,circle,thick] (85) at (9, 2) {};
		\node [draw,dotted,circle,thick] (86) at (10, 2) {};
		\node [draw,dotted,circle,thick] (87) at (12, 2) {};
		\node  (88) at (0, 5.5) {};
		\node  (89) at (0, 4.5) {};
		\node  (90) at (12, 5.5) {};
		\node  (91) at (12, 4.5) {};
		\draw [style=new edge style 0] (48) to (56.center);
		\draw [style=new edge style 0] (48) to (57.center);
		\draw [style=new edge style 0] (57.center) to (58.center);
		\draw (63.center) to (61.center);
		\draw (61.center) to (62.center);
		\draw (62.center) to (63.center);
		\draw [dotted] (64.center) to (65.center);
		\draw (61.center) to (67.center);
		\draw (67.center) to (62.center);
		\draw [style=new edge style 0] (75.center) to (76.center);
		\draw [style=new edge style 0] (77.center) to (78.center);
		\draw (77.center) to (48);
		\draw (48) to (75.center);
		\draw (88.center) to (89.center);
		\draw (90.center) to (91.center);
\end{tikzpicture}

        \end{center}

        \caption{\label{AnABnAn} The sets produced by $A^{k+1} B^k A^k$.}
    \end{figure}
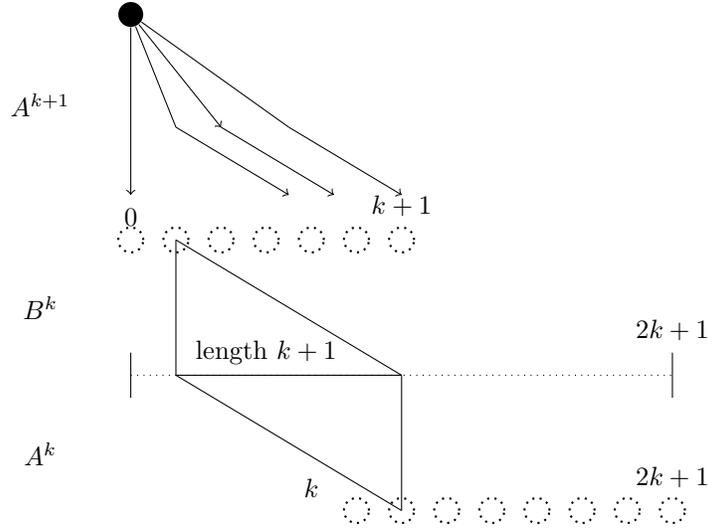

    Now we prove that every set obtained by reading  $A^k B^k A^{k+1}$ has 
    at least one of $k, \ldots, 2k+1$.
    After reading $A^k B^k$ we have any interval of size $k+1$ between $0$ and $2k$. The state $k$ is always in the interval,
    and must go somewhere between positions $k$ and $2k+1$ after $A^{k+1}$.

    %
    %
    %
    %

    As for \ref{chain4}, reading $A^{n-1}$ or $A^n$ in any singleton game state $\{q\}$ produces exactly the states $\{q\}, \{q+1\}, \ldots, \{m-1\}$ and the loop $\{m\}, \ldots, \{m+p-1\}$.
    Likewise, reading $B^{n-1}$ or $B^n$ produces $\{q, q+1, \ldots, m-1\} \cup \{m, \ldots, m+p-1\}$.
  \end{proof}

\begin{theorem}
  \label{chainBoundedBound}
  Let $\mathcal{A}$ be a $1$-bounded chain automaton with $n$ states.
  The number of states in the minimal DFA of $W(\mathcal{L}(\mathcal{A}))$ is $O(n^{1/5} e^{4 \pi \sqrt{\frac{n}{3}}})$.
\end{theorem}
\begin{proof}
  Since $\mathcal{A}$ does not accept any word with $n$ or more $1$-symbols, $W(\mathcal{L}(\mathcal{A}))$ contains no word with $n$ or more $B$-symbols.
  The equivalences $B^k A^k B^{k+1} \equiv B^{k+1} A^k B^k$ and $A^{k+1} B^k A^k \equiv A^k B^k A^{k+1}$ for $k \geq 0$
  that follow from Lemma~\ref{chainEquivWords}
  allow us to rewrite every word of $W(\mathcal{L}(\mathcal{A}))$ in the form $A^{n_1} B^{n_2} A^{n_3} B^{n_4} \cdots A^{n_{2r-1}} B^{n_{2r}}$ where
  the sequence $n_1, \ldots, n_{2r}$ is first nondecreasing and then nonincreasing, and $n_2 + n_4 + \cdots + n_{2r} < n$.
  With Lemma~\ref{chainEquivWords}\ref{chain4} we can also guarantee $n_1, n_3, \ldots, n_{2r-1} < n$, so that $\sum_i n_i < 4 n$.
  In \cite{auluck1951some}, Auluck showed that the number $Q(m)$ of partitions $m = n_1 + \ldots n_r$ of an integer $m$ that are first nondecreasing and then nonincreasing is $\Theta(m^{-4/5} e^{2 \pi \sqrt{m/3}})$.
  Of course, $v \equiv w$ implies $v \sim w$.
  Thus the number of non-right-equivalent words for $W(\mathcal{L}(\mathcal{A}))$, and the number of states in its minimal DFA, is at most $1 + \sum_{m = 0}^{4 n-1} Q(m) = O(n^{1/5} e^{4 \pi \sqrt{\frac{n}{3}}})$.
\end{proof}

\section{Case study: exact number of $1$-symbols}

  In the previous section we proved a bound for the complexity of the winning set of a bounded permutation invariant language. Here we study 
  a particular case, the language of words with exactly $n$ ones, or $L= (0^*1)^n 0^*$.
  We not only compute the number of states in the minimal automaton (which is cubic in $n$), but also describe the winning set.
  Throughout the section $\mathcal{A}$ is the minimal automaton for $L$, described in Figure \ref{01s01s}.
  For $S \subseteq Q$, we denote $\overline{S} =  \{\min(S), \min(S) + 1, \ldots, \max(S)\}$, and for any game state $\gst{S}$ of $W(\mathcal{A})$, denote $\overline{\gst{S}} = \{ \overline{S} \;|\; S \in \gst{S}\}$.

  \begin{figure}[htp]
    
    \begin{center}
      \begin{tikzpicture}
	[nonaccept/.style={circle,draw,inner sep=0cm,minimum size=0.85cm},
	 accept/.style={circle,draw,double,inner sep=0cm,minimum size=0.85cm},
	 every edge/.style={draw,->,>=stealth},
	 scale=0.5]
		\node (0) at (-9, 0) {};
		\node [nonaccept] (5) at (-7, 0) {$0$};
		\node [nonaccept] (6) at (-3, 0) {$1$};
		\node (7) at (0, 0) {$\dots$};
		\node [nonaccept] (8) at (3, 0) {$n-1$};
		\node [accept] (9) at (7, 0) {$n$};
		\node [nonaccept] (11) at (11, 0) {$n+1$};
		\draw (0) edge (5);
		\draw (5) edge (6);
		\draw (6) edge (7);
		\draw (7) edge (8);
		\foreach \k in {5,6,8,9,11}{
			 \draw [in=55, out=125, loop, looseness=4] (\k) edge (\k);
		}
		\draw [in=40, out=140, loop] (11) edge (11);
		\draw (8) edge (9);
		\draw (9) edge (11);
\end{tikzpicture}

    \end{center}
    
    \caption{\label{01s01s} The minimal DFA for $L= (0^*1)^n0^*$.}
  \end{figure}
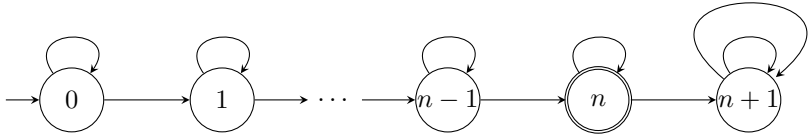

  \begin{lemma} \label{lem:stateIsInterval}
    Each game state $\gst{S}$ of $W(\mathcal{A})$ is equivalent to $\overline{\gst{S}}$.
  \end{lemma}

    \begin{proof}
    We prove by induction that for every
    $w\in \{A,B\}^*$, the game state $\delta_W(\gst{S},w)$ is final iff $\delta_W ( \overline{\gst{S}},w)$ is.
    \begin{itemize}
        \item  For the empty word $\lambda$, $\delta_W (\gst{S}, \lambda) = \gst{S}$ is final
        iff $\{n\} \in \gst{S}$
        iff $\{n\} \in \overline{\gst{S}} = \delta_W(\overline{\gst{S}}, \lambda)$.
        \item For $w= B v$,
          it's easy to see that $ \delta_W (\overline{\gst{S}}, B) = \overline{ \delta_W (\gst{S},B)}$.
          By the induction hypothesis, $ \delta_W (\gst{S}, B v) = \delta(\delta_W (\gst{S}, B), v)$ is final iff
          $ \delta_W (\overline{\delta_W (\gst{S}, B)}, v) = \delta_W (\delta_W (\overline{\gst{S}}, B), v) = \delta_W (\overline{\gst{S}}, B v)$ is.
        \item For $w= A v$,
        we have
         $ \overline{\delta_W (\overline{\gst{S}}, A)} = \overline{ \delta_W (\gst{S}, A)}$: 
        For each set $S \in \gst{S}$, we focus on a set $R$ that can be obtained
        by a combination of choices by reading one $A$.
        The set $\overline{R}$ is only defined by the leftmost and rightmost elements in $R$.
        From $\overline{S}$, by making the same choices for the leftmost and rightmost
        elements we can produce the same leftmost and rightmost elements as in $R$ to obtain a set equivalent to $\overline{R}$.

        By the induction hypothesis, the game state $\delta_W (\gst{S}, A v) = \delta_W (\delta_W (\gst{S}, A), v)$ is final iff
        $\delta_W (\overline{\delta_W (\gst{S}, A)}, v) = \delta_W (\overline{\delta_W (\overline{\gst{S}}, A)}, v)$ is.
        Again by the induction hypothesis, this is equivalent to $\delta_W (\delta_W (\overline{\gst{S}}, A), v) = \delta_W(\overline{\gst{S}}, A v)$ being final.
    \end{itemize}
\end{proof}

  \begin{lemma} \label{lem:formOfReachables}
    Let $T$ be the set of integer triples $(i, \ell, N)$ with $0 \leq i \leq n$, $1 \leq \ell \leq n-i+1$ and $1 \leq N \leq n-i-\ell+2$.
    For $(i, \ell, N) \in T$, let
    \[ \gst{S}(i, \ell, N) :=
      \{
      \{i, \ldots, i+\ell-1\},
      \{i+1, \ldots, i+\ell\},
      \ldots,
      \{i+N-1, \ldots, i+\ell+N-2\} \}.
    \]
    \begin{roster}
    \item \label{form1}
      Each reachable game state of $W(\mathcal{A})$ is equivalent to some $\gst{S}(i, \ell, N)$ for $(i, \ell, N) \in T$, or to $\emptyset$.
    \item \label{form2} The game states $\gst{S}(i,\ell,N)$ for $(i, \ell, N) \in T$ are nonequivalent.
    \item \label{form3} Every $\gst{S}(i,\ell,N)$ for $(i, \ell, N) \in T$ is equivalent to some reachable game state.
    \end{roster}
  \end{lemma}

\begin{proof}
  We first make the following remarks which follow from Lemma~\ref{lem:stateIsInterval}.
  Let $S = \{ i, \ldots, j\}$. If $i=j$, then $\delta_W (\{S\}, A) \sim \{ \{i\}, \{i+1\} \}$,
  otherwise $\delta_W (\{S\}, A) \sim \{ \{i+1, \dots, j\} \}$.
  In both cases we also have $\delta_W (\{S\}, B) \sim \{ \{i, \dots, j+1\} \}$.
  
  We prove~\ref{form1}.
  From $\gst{S}(i,\ell,N)$, if we read $B$, we have $\gst{S}(i,\ell+1,N)$. If we read $A$ and $\ell=1$, we get $\gst{S}(i,1,N+1)$. If we read $A$ and $\ell > 1$, we obtain $\gst{S}(i+1,\ell-1,N)$.
  The claim follows since the initial game state is $\gst{S}(0, 1, 1)$, and if a set in the game state contains the state $n+1$, it is equivalent to $\emptyset$.
  
  For~\ref{form2}, we first distinguish game states with different lengths $\ell<\ell'$.
  By reading $A^\ell (B A)^k$ for a suitable $k \geq 0$ we reach a final game state from $\gst{S}(i, \ell, N)$ but not from any $\gst{S}(i', \ell', N')$ with $\ell < \ell'$.
  
  Now we suppose we have game states $\gst{S}(i,\ell,N)$ and $\gst{S}(i',\ell,N')$.
  By reading $A^{\ell-1}$, we get respectively $\gst{S}(i+\ell-1,1,N)$
  and $\gst{S}(i'+\ell-1,1,N')$, which are intervals of singleton sets. Since 
  these distributions of singleton sets are different, we can read $(B A)^k$ 
  for a suitable $k \geq 0$ to obtain a final game state from one of them but not the other.
  
  For~\ref{form3}, reading $(B A)^i A^{N-1} B^{\ell-1}$ leads to a game state equivalent to $\gst{S}(i,\ell,N)$.
\end{proof}

The game state $\gst{S}(i,\ell,N)$ is an interval of intervals, where $i$ is the leftmost position of the first interval, $\ell$ is the common length of the intervals, and $N$ is the number of intervals.
It is easy to see that reading $B A$ from $\gst{S}(i,\ell,N)$ produces a game state equivalent to $\gst{S}(i+1,\ell,N)$.
In other ords, $BA$ can be use to ``make the game state go forward'' in the chain automaton.

\begin{proposition} \label{lem:formOfH}
  The minimal automaton for $W(L)$ has $\frac{n^3}{6}+n^2+\frac{11n}{6}+2$ states.
\end{proposition}

\begin{proof}
  The parameter $i$ can vary from $0$ to $n$, $\ell$ from $1$ to $n-i+1$, and $N$ from $1$ to $n-i-\ell + 2$.
  We need to consider one more state for the sink state $\emptyset$. In total, there are $\frac{n^3}{6}+n^2+\frac{11n}{6}+2$ states.
\end{proof}

\begin{proposition}
  $W(L)$ is exactly the set of words $w \in \{A, B\}^*$ such that $|w|_A \geq n$, $|w|_B \leq n$, and every suffix $v$ of $w$ satisfies $|v|_A \geq |v|_B$.
\end{proposition}

\begin{proof}
    Every word of $W(L)$ has at least $n$ letters $A$ to let Alice win against Bob when he only plays $0$s.
    Similary it must have a most $n$ letters $B$ to let Alice wins when Bob only plays $1$s.
    Only game states of the form $\gst{S}(i,1,N)$ can be accepting.
    Since doing $A$ decreases the parameter $\ell$ by one,
    and $B$ increases $\ell$ by one, words of $W(L)$ must have after each $B$ an associated $A$ somewhere in the word.
    This is equivalent to the suffix condition.

    Conversely, if a word $w$ has this property, since $B A$ makes the whole game state move forward along the chain automaton,
    we can move all occurrences of $B A$ to the beginning of the word to obtain $w'$, which is in $W(L)$ iff $w$ is.
    Then $w' = (B A)^{k_1} A^{k_2} B^{k_3}$ for some $k_1, k_2, k_3 \geq 0$. Because of the suffix condition, there are no $B$ at the end of $w'$, so $k_3 = 0$.
    We also have $|w'|_A = |w|_A \geq n$ and $|w'|_B = |w|_B \leq n$, so $k_1+k_2 \geq n$ and $k_1 \leq n$.
    This means $\delta_W(\gst{S}(0,1,1), w') = \gst{S}(k_1, 1, k_2)$ so we have a game state with singletons and one of them is at position $n$, hence it is final and $w', w \in W(L)$.
\end{proof}

\section{A context-free language}

In this section we prove that the winning set operator does not in general preserve context-free languages by studying the winning set of the Dyck langage. For better readability, $0$ stands for the opening parenthesis $($ and $1$ stands for the closing parenthesis $)$.

\begin{proposition}
  \label{DyckWinningNotCFL}  
  Denote by $D$ the Dyck language.
  The winning shift $W(D)$ is not context-free.
\end{proposition}
\begin{proof}
    Suppose by contradiction that $W(D)$ is context-free.
    Take $L = W(D) \cap (AA)^*(BB)^*(AA)^*$ which is context-free by intersection of a regular language and a context-free language.
    
    We claim that
    $L =  \{ A^{2i} B^{2j} A^{2k} \;|\; i \geq j, k\geq 2j \}$.
    First, when Bob plays $2j$ times in a row, he can close $2j$ parentheses.
    Alice must play at least $2j$ times before that and open at least
    $2j$ parentheses in order to have a chance to win.
    But then Bob can open $2j$ parentheses instead of closing them, which means that when Alice plays a second time, she has to be able to close $4j$ parentheses.
    This means the right hand side contains $L$.

    Let then $w = A^{2i} B^{2j} A^{2k}$ with $i \geq j$ and $k \geq 2j$.
    A winning strategy for Alice on $w$ is to do first play ``01'' $i-j$ times and then ``0'' $2j$ times so that there are $2j$ parentheses left to be closed.
    After Bob plays $2j$ times, there is an even number $2h$ of parenthesis to be closed, smaller than $4j$.
    Then Alice wins by playing ``1'' $2h$ times and ``01'' $k-h$ times, which is legal because $h \leq 2j \leq k$.

    We apply Ogden's lemma on $L$ to show that it's not context-free.
    Take $N$ obtained from the lemma.
    Consider the the word $w = A^{2N} B^{2N} A^{4N} \in W(L)$ where we mark the $2N$ occurrences of $B$.
    By the lemma, $w$ can be written as $xuyvz$ so that $x u^n y v^n z \in L$ for all $n \geq 0$ and ($x$ and $u$ and $y$) or ($y$ and $v$ and $z$) contains at least one marked position.
    Because $u$ and $v$ can be repeated, one of them only contains $B$s and the other contains only $A$s or only $B$s.
    But then in the words $x u^n y v^n z$, as $n$ increases, the number of $A$s in the left part or in the right part remains constant while the number of $B$s increases.
    For large enough $n$ this contradicts $x u^n y v^n z \in L$.
  \end{proof}
  
\bibliographystyle{plain}
\bibliography{references}

\end{document}

%% file: tikzstyle.tikzstyles
\tikzstyle{new style 0}=[fill=white, draw=black, shape=circle, minimum size=1cm, node distance=2cm, font={\small}]
\tikzstyle{small}=[fill=white, draw=black, shape=circle, minimum size=0.3cm, node distance=2cm, font={\small}]
\tikzstyle{dot}=[fill=black, draw=black, shape=circle, minimum size=0.1pt, node distance=2cm, font={\small}]
\tikzstyle{accept}=[fill=white, draw=black, shape=circle, accepting, minimum size=1cm, tikzit draw={rgb,255: red,61; green,195; blue,12}]
\tikzstyle{dottednode}=[fill=white, dotted, draw=black, shape=circle, line width=1.3, minimum size=0.1pt, node distance=2cm, font={\small}, tikzit fill={rgb,255: red,154; green,154; blue,154}]
\tikzstyle{new edge style 1}=[--]
\tikzstyle{greydottednode}=[fill={rgb,255: red,172; green,172; blue,172}, dotted, draw=black, shape=circle, line width=1.3, minimum size=0.1pt, node distance=2cm, font={\small}, tikzit fill={rgb,255: red,194; green,194; blue,194}, tikzit draw={rgb,255: red,102; green,102; blue,102}]
\tikzstyle{new style 1}=[fill={rgb,255: red,49; green,255; blue,35}, draw=black, shape=circle]

\tikzstyle{new edge style 0}=[->]
\tikzstyle{doublearr}=[->,double]
\tikzstyle{dottedline}=[-, dotted, line width=1.5, tikzit draw={rgb,255: red,99; green,99; blue,99}]
\tikzstyle{new edge style 1}=[-, draw={rgb,255: red,26; green,50; blue,189}]
\tikzstyle{new edge style 2}=[-, draw={rgb,255: red,33; green,126; blue,27}]
\tikzstyle{new edge style 3}=[->, draw={rgb,255: red,245; green,60; blue,54}]
\tikzstyle{bluearrow}=[draw={rgb,255: red,25; green,31; blue,148}, ->]
\tikzstyle{greenarrow}=[->, draw={rgb,255: red,55; green,112; blue,18}]